\documentclass[a4paper,onecolumn,10pt,peerreview,draftcls]{IEEEtran}
\usepackage{amsmath}
\usepackage{amssymb}
\usepackage{graphicx}
\usepackage{subfigure}
\usepackage{algorithm}
\usepackage{algorithmic}
\usepackage{multicol}
\usepackage{moreverb}

\newcommand\BibTeX{{\rmfamily B\kern-.05em \textsc{i\kern-.025em b}\kern-.08em
T\kern-.1667em\lower.7ex\hbox{E}\kern-.125emX}}
\renewcommand{\baselinestretch}{1.5}
\graphicspath{{./figures/}}

\newtheorem{theorem}{Theorem}[section]
\newtheorem{lemma}{Lemma}

\begin{document}

\title{A Primal-Dual Algorithm for a Heterogeneous Traveling Salesman Problem}

\author{Jungyun Bae$^1$, Sivakumar Rathinam$^2$
\thanks{1. Graduate Student, Department of Mechanical Engineering, Texas A
\& M University, College Station, Texas, U.S.A 77843.}
\thanks{2. Assistant Professor,
Department of Mechanical Engineering, Texas A \& M University,
College Station, Texas, U.S.A 77843.}
}
\maketitle

\begin{abstract}
Surveillance applications require a collection of heterogeneous vehicles to visit a set of targets. We
consider a fundamental routing problem that arises in these applications involving two vehicles. Specifically, we consider a routing problem where there are two heterogeneous vehicles that start from distinct initial locations, and a set of targets. The objective is to find a tour for each vehicle such that each of the targets is visited at least once by a vehicle and the sum of the distances traveled by the vehicles is a minimum. We present a primal-dual algorithm for a variant of this routing problem that provides an approximation ratio of 2.
\end{abstract}

{}

\section*{INTRODUCTION}

Heterogeneous unmanned vehicles are commonly used in surveillance applications for monitoring and tracking a set of targets. For example, in the Cooperative Operations in Urban Terrain project \cite{COUNTER} at the Air Force Research Laboratory, a team of unmanned vehicles are required to monitor a set of targets and send information/video about the targets to the ground station controlled by a human operator. The human operator may further add new locations of potential targets or task the vehicles to revisit the targets at different angles. Once the human operator enters his/her input through the human-machine interface, the central computer associated with the interface has few minutes to determine the motion plans for each of the vehicles. A fundamental subproblem that has to be solved by this computer is the problem of finding a tour for each vehicle so that each target is visited at least once by some vehicle and an objective that depends on the distances traveled by the vehicles is a minimum. A common objective that is used for these applications is the sum of the total distances traveled by all the vehicles. If there is only one vehicle, this routing problem is referred to as the Traveling Salesman Problem (TSP) in the literature. If there are multiple vehicles that (possibly) start from different initial locations or depots, then this routing problem is referred to as the Multiple Depot, TSP. Once the routing problem is solved and the tours have been determined, a nominal trajectory can be specified for each vehicle to include other kinematic constraints of the vehicles using the results in \cite{Lee:5approx}, \cite{Rathinam_IEEEAutomation}.

A multiple depot, TSP is a generalization of the single TSP and is NP-Hard. This routing problem is further complicated if the vehicles involved are heterogeneous. In this article, vehicles are considered to be heterogeneous if the distance to travel between any two targets depend on the type of the vehicle used. In the context of unmanned applications, as a multiple depot heterogeneous TSP is generally a subproblem that needs to be solved, we are interested in developing fast algorithms that produce approximate solutions than find optimal solutions that may be relatively difficult to solve. Therefore, the main focus of this article is to develop approximation algorithms for heterogeneous TSPs. An approximation algorithm for a problem is an algorithm that runs in polynomial time and produces a solution whose cost is at most a given factor away from the optimal cost for every instance of the problem.

The objective of this article is to develop a primal-dual algorithm for a two depot, heterogeneous TSP (2DHTSP). In addition to assuming that the costs satisfy the triangle inequality for each vehicle, we consider a variant of the problem where the cost of traveling between any two targets for the first vehicle is at most equal to the cost of traveling between the same targets for the second vehicle. Using these assumptions, we show that the developed primal-dual algorithm has an approximation ratio of 2. We are motivated to address this variant of the 2DHTSP due to the following reasons:
\begin{enumerate}
\item The 2DHTSP is one the simplest cases of the general multiple depot, heterogeneous TSP. The objective of this work is to first develop good algorithms that can handle these simpler cases efficiently.
\item Consider a scenario where each of the vehicles is modeled as a ground robot that can move both forwards and backwards with a constraint on its minimum turning radius\cite{reed_shepp}. If the approach angle at each target is given and the minimum turning radius of the first vehicle is at most equal to the minimum turning radius of the second vehicle, it follows that the optimal distance required to travel between any two targets for the first vehicle will be at most equal to the optimal distance required for the second vehicle. Therefore, the problem addressed in this article is a useful variant to address.
 \item The 2DHTSP is a generalization of a 2 depot, homogeneous TSP where there are additional \textit{vehicle-target} constraints which require one of the vehicles to necessarily visit a given subset of targets in addition to visiting any common target available for both the vehicles. This variant of 2 depot, homogeneous TSP arises in applications where the distance to travel between the targets are identical for both the vehicles, but one of the vehicles carry sensors that require the vehicle to visit a subset of targets compulsorily.
  \item For some missions involving identical vehicles, it is sometimes necessary to minimize the maximum cost incurred by any of the vehicles. This problem is referred to as the min-max, multiple depot, homogeneous TSP in the literature. If there are only two vehicles involved, one can use the variant of the heterogeneous TSP considered in this article to compute bounds for the min-max problem. Specifically, let $TOUR_1$ and $TOUR_2$ denote a feasible pair of tours for the first and the second vehicle respectively. Also, for $i=1,2$, let $cost(TOUR_i)$ denote the cost of traversing the tour for the $i^{th}$ vehicle. Then, the min-max problem can be formulated as $\min_{TOUR_1,TOUR_2} z$ subject to the constraints $cost(TOUR_1)\leq z$, and $cost(TOUR_2)\leq z$. By dualizing the constraints, one obtains a relaxed problem of the form $\max_{\pi_1+\pi_2=1} \min_{TOUR_1,TOUR_2}~ [\pi_1cost(TOUR_1) + \pi_2 cost(TOUR_2)]$. Therefore, for a given value of the penalty variable $\pi_1$, the relaxation involves solving the heterogeneous TSP considered in this article.
\end{enumerate}

Without the assumptions on the costs of the two vehicles, the 2DHTSP is a generalization of the standard variant of the prize collecting TSP considered by Goemans and Williamson in \cite{GoemansW95}. In this variant, each target essentially has a penalty associated with it. The objective of the prize collecting TSP is to find a tour for the vehicle that starts and ends at the depot such that the cost of the tour plus the sum of the penalties of each target not present in the tour is a minimum. For any two vertices $i$ and $j$, if $\pi_i, \pi_j$ denote the penalties of $i$ and $j$ respectively, then one can pose the prize collecting TSP as a 2DHTSP by setting the cost of traveling the edge joining vertices $i$ and $j$ for the second vehicle to be equal to $\frac{\pi_i+\pi_j}{2}$. Essentially, by choosing the penalty variable corresponding to the second depot to be equal to 0, one can deduce that the travel cost for the second vehicle is actually equal to the sum of the penalties of the targets not present in the tour of the first vehicle. Even though there are no penalties explicitly mentioned in the 2DHTSP, the tour cost for the second vehicle which essentially account for targets not visited by the first vehicle act as penalties. Essentially, our algorithm is based on the well known moat growing procedure proposed by Goemans and Williamson in \cite{GoemansW95}. For these reasons, the primal-dual algorithm presented in this article is based on the primal-dual algorithm available for the prize-collecting TSP in \cite{GoemansW95}.

Most of the work in the literature related to approximation algorithms for multiple depot, TSPs deal with identical vehicles. For example, when the costs satisfy the triangle inequality, there are several approximation algorithms for the multiple depot, homogeneous TSP in \cite{Rathinam_IEEEAutomation},\cite{approxMDMTSPmalik},\cite{approx2HPP},\cite{XuZhou}. Recently, a $3-$approximation algorithm was presented for a two depot, heterogeneous TSP in \cite{sai-heterogeneousTSP}. This algorithm partitions the targets by solving a linear programming relaxation and then uses Christofides algorithm \cite{christofides} to find a sequence of targets for each vehicle.

The 2-approximation algorithms available in the literature for the multiple depot, TSP generally follow a two-step procedure. In the first step, a constrained forest problem which is generally a relaxation of the multiple depot, TSP is solved optimally. In the second step, an Eulerian graph is found for each vehicle based on the constrained forest. From these Eulerian graphs, a tour can be found for each vehicle by short-cutting any target already visited by a vehicle. In this article, we follow a similar procedure where we first find a heterogeneous spanning forest using a primal-dual algorithm by solving a relaxation of the 2DHTSP. Then, the edges in the heterogeneous spanning forest are doubled to obtain an Eulerian graph for each vehicle. Given these Eulerian graphs, one can \cite{approxalgo} always find a tour for each vehicle that visits each of the targets exactly once. {\it The crux of this procedure depends on finding a good heterogeneous spanning forest}. Using a primal-dual algorithm, we find a heterogeneous spanning forest whose cost is at most equal to the optimal cost of the 2DHTSP in polynomial time. Hence, it follows that the approximation ratio of the proposed procedure is 2.

\section{Problem Statement}\label{sec:problem}
Let $D=\{d_1,d_2\}$ denote the two depots (initial locations) corresponding to the first and the second vehicle respectively. Let $T$ be the set of targets to be visited by both the vehicles. Let $V_1:=T\bigcup \{d_1\}$ be the set of vertices corresponding to the first vehicle. Similarly, let $V_2:=T\bigcup \{d_2\}$ be the set of vertices corresponding to the second vehicle. For $i=1,2$, let $E_i$ denote the set of all the edges that join any two distinct vertices in $V_i$. Let the cost of traversing an edge $e\in E_1$ for the first vehicle be denoted by $cost^1_{e}$. Similarly, let the cost of traversing an edge $e \in E_2$ for the second vehicle be denoted by $cost^2_{e}$. We will assume that it is always cheaper to travel between any two targets using the first vehicle as compared to using the second vehicle, $i.e.$, for any edge $e$ joining two targets, $cost^1_{e}\leq cost^2_{e}$. We also assume that the costs satisfy the triangle inequality for both the vehicles.

A tour for the first vehicle starts from its depot $d_1$, visits a set of targets in a sequence and finally returns to $d_1$. A tour for the second vehicle starts from its depot $d_2$, visits a set of targets in a sequence and finally returns to $d_2$. The objective of the 2DHTSP is to find a tour for each vehicle such that each target is visited exactly once by some vehicle and the sum of the cost of the edges traveled by both the vehicles is a minimum.

\section{Problem formulation}
Let $x_e$ be an integer variable that represents whether edge $e\in E_1$ is present in the tour corresponding to the first vehicle. For any edge $e$ joining \textit{two targets}, $x_e$ can take values only in the set $\{0,1\}$; $x_e=1$ if $e$ is present in the tour of the first vehicle and $x_e=0$ otherwise. In order for a tour to visit just one target if required, $x_e$ is allowed to choose any of the values in the set $\{0,1,2\}$ for an edge $e$ joining the depot $d_1$ and a target $v\in T$. Similarly, let $y_e$ be an integer variable that represents whether edge $e\in E_2$ is present in the tour corresponding to the second vehicle. Let $z_U$ be a binary variable that determines the partition of targets connected to the first and the second depot; $z_U$ is equal to 1 if each target in $U\subseteq T$ is connected to the second depot and each target in $T\setminus U$ is connected to the first depot. There is at most one subset of targets, $U$, that is allowed to have $z_U$ to be equal to 1. Let $\delta_i(S)$ (for $i=1,2$) denote the subset of all the edges of $E_i$ with one end in $S$ and an other end in $V_i\setminus S$. $\delta_i(S)$ is also referred to as the cut set of $S$ corresponding to the $i^{th}$ vehicle.

For any $S\subseteq T$, at least two edges must be chosen from $\delta_1(S)$ for the tour of the first vehicle if there is at least one vertex in $S$ that is not connected to the second depot, $i.e.$, $\sum_{e\in \delta_1(S)} x_e \geq 2$ if $\sum_{T\supseteq U\supseteq S} z_U = 0$. This requirement can be written as $\sum_{e\in \delta_1(S)} x_e  + 2\sum_{T\supseteq U\supseteq S} z_U  \geq 2$. Similarly, for any $S\subseteq T$, at least two edges must be chosen from $\delta_2(S)$ for the tour of the second vehicle if all the vertices in $S$ are required to be visited by the second vehicle. This requirement can be expressed as $\sum_{e\in \delta_2(S)}y_e  \geq 2\sum_{T\supseteq U\supseteq S} z_U$. Now, consider the following integer programming relaxation for the 2DHTSP without the degree constraints:

\begin{align*}
C_{lp} = \min \sum_{e \in E_1} cost_e^1 ~ x_e & + \sum_{e \in E_2} cost_e^2 ~ y_e
\end{align*}

\begin{align}
\sum_{e\in \delta_1(S)} x_e  + 2\sum_{T\supseteq U\supseteq S} z_U  &\geq 2 \hspace{1.9cm} \forall  S \subseteq T, \label{eq1} \\
\sum_{e\in \delta_2(S)}y_e & \geq 2\sum_{T\supseteq U\supseteq S} z_U \quad   \forall  S \subseteq T, \label{eq2} \\
\sum_{U \subseteq T} z_U & \leq 1,  \label{eq5} \\
x_e,y_e  \in \{0,1\} &~ \forall e\textrm{ joining any two targets}, \label{eq6} \\
x_e  \in \{0,1,2\} &~ \forall e\textrm{ joining }d_1\textrm{ and a target}, \\
y_e  \in \{0,1,2\} &~ \forall e\textrm{ joining }d_2\textrm{ and a target}, \\
 \quad z_U  \in \{0,1\} &~ \forall U \subseteq T. \label{eq7}
\end{align}

Consider a Linear Programming (LP) relaxation of the above integer program where the constraints (\ref{eq5})-(\ref{eq7}) are relaxed as follows:
\begin{align}
C_{lp} = \min \sum_{e \in E_1} cost_e^1 ~ x_e & + \sum_{e \in E_2} cost_e^2 ~ y_e
\end{align}

\begin{align}
\sum_{e\in \delta_1(S)} x_e  + 2\sum_{T\supseteq U\supseteq S} z_U  &\geq 2 \hspace{1.9cm} \forall  S \subseteq T, \label{eq:TSPvehicle1}\\
\sum_{e\in \delta_2(S)}y_e & \geq 2\sum_{T\supseteq U\supseteq S} z_U \quad   \forall  S \subseteq T, \label{eq:TSPvehicle2}\\
x_e  \geq 0 ~ \forall e & \in E_1,\quad y_e \geq 0 ~ \forall e  \in E_2, \nonumber \\
 \quad z_U & \geq 0 ~ \forall U \subseteq T.
\end{align}

\newpage

A dual of the above LP relaxation can be formulated as follows:
\begin{align}
C_{dual} = \max~ 2 \sum_{S \subseteq T} Y_1(S)
\end{align}

\begin{align}
\sum_{S: e\in \delta_1(S)} Y_1(S)  & \leq cost^1_e \hspace{1.1cm} \forall  e \in E_1, \label{eq:dualfirst}\\
\sum_{S: e\in \delta_2(S)} Y_2(S)  & \leq cost^2_e \hspace{1.1cm} \forall  e \in E_2, \label{eq:dualsecond}\\
\sum_{S \subseteq U} Y_1(S)  & \leq \sum_{S \subseteq U} Y_2(S)  \quad \forall U \subseteq T, \label{eq:dualbound}\\
Y_1(S) ,Y_2(S) & \geq 0 \hspace{1.7cm} \forall S \subseteq T.
\end{align}

We use the above dual problem to find a Heterogeneous Spanning Forest (HSF). A HSF is a collection of two trees  where the first tree spans a subset of targets and $d_1$, and the second tree connects the remaining set of targets to $d_2$. In the next section, we discuss the main ideas involved in the primal-dual algorithm that finds a HSF. We later present the details of the algorithm and show that the cost of this HSF is at most equal to the optimal cost of the above dual. This leads to a 2-approximation algorithm for the 2DHTSP.

\section{Main ideas of the Primal Dual Algorithm}
The primal-dual algorithm follows the greedy procedure outlined by Goemans and Williamson in \cite{GoemansW95}. The basic structure of the algorithm involves maintaining a forest of edges corresponding to each vehicle, and a solution to the dual problem. The edges in the forests are candidates for the set of edges that finally appear in the output (HSF) of the algorithm. Suppose $F_1$ and $F_2$ denote the forest corresponding to the first and the second vehicle respectively. Let the set of connected components in $F_1$ and $F_2$ be denoted by ${\mathcal{C}}_1$ and ${\mathcal{C}}_2$ respectively. Initially, both ${\mathcal{C}}_1$ and ${\mathcal{C}}_2$ consist of components where each vertex is in its own connected component, $i.e.$, ${\mathcal{C}}_1=\{\{v\}: v \in V_1\}$ and ${\mathcal{C}}_2=\{\{v\}: v \in V_2\}$. That is, both $F_1$ and $F_2$ are empty. All the components are initially active except the components that contain the depots (Refer to the figures \ref{step1}-\ref{step8} for an illustration of the algorithm). Also, all the dual variables are set to zero.

 The primal-dual algorithm is an iterative algorithm where in each iteration, at most one edge is added between two distinct components of $F_1$ or $F_2$ thus merging the two components. The choice of selecting the appropriate edge to be added is based on a dual solution which is also updated during each iteration. Specifically, in each iteration, the algorithm uniformly increases the dual variable of each active component by a value that is as large as possible such that none of the constraints in (\ref{eq:dualfirst})-(\ref{eq:dualbound}) are violated. When the dual variables are increased, one of the following outcomes is possible:

 \begin{itemize}
 \item If any of the constraints in (\ref{eq:dualfirst})-(\ref{eq:dualsecond}) becomes tight for some edge $(u,v)\in E_i,~ i=1,2$ between two distinct components in $F_i$, then the algorithm adds $(u,v)$ to $F_i$ and merges the two components (Refer to figures \ref{step2}-\ref{step4}). If the merged component contains a depot, it becomes inactive; otherwise it is active.  We can also explain this outcome in the following way: Suppose $p_i(u):=\sum_{S: u\in S} Y_i(S)$ is the total price all the components containing target $u$ are willing to pay to develop a network $F_i$ that can connect $u$ to depot $d_i$. Then, the edge $e:=(u,v)$ is added to $F_i$ when $p_i(u) + p_i(v)=cost^i_e$, $i.e.$, the price paid by the components containing $u$ and the components containing $v$ equals the cost of adding an edge $(u,v)$ to the network. If a component $\overline{C}$ of $F_2$ merges with the depot $d_2$ (figure \ref{step4}), then $\overline{C}\bigcup \{d_2\}$ becomes inactive, and the total price $\sum_{S \subseteq \overline{C}} Y_2(S)$ serves as an upper bound for $\sum_{S \subseteq \overline{C}} Y_1(S)$.

 \item If a constraint in (\ref{eq:dualbound}) becomes tight for a component $\overline{C}$, then $\overline{C}$ is deactivated in $F_1$ (Refer to figures \ref{step5},\ref{step7}). This outcome occurs when the total price ($\sum_{S \subseteq \overline{C}} Y_1(S)$) that all the vertices in $\overline{C}$ are willing to pay to get connected to $d_1$ becomes as costly as the total price ($\sum_{S \subseteq \overline{C}} Y_2(S)$) that the same vertices have already paid to get connected to $d_2$.
\end{itemize}

  The iterative process terminates when all the components become inactive. The final step of the algorithm removes any unnecessary edges (refer to figure \ref{step8}) that are not required to be in $F_1$ or $F_2$ using a marking procedure that was previously used for the prize collecting TSP in \cite{GoemansW95}. \\

  There is a key feature to note in our primal-dual procedure. We ensure that the dual variables of all the active components in $F_1$ and $F_2$ are increased uniformly by the same amount in each iteration. The components in $F_1$ tend to merge first as compared with the components in $F_2$ due to the choice of our dual increase and the fact that it is cheaper to travel between any two targets using the first vehicle as compared with the second vehicle. As the algorithm progresses, for any $U \subseteq T$, it is likely that $\sum_{S \subseteq U} Y_1(S) < \sum_{S \subseteq U} Y_2(S)$ as there may be fewer active components of $F_1$ in $U$ as compared to $F_2$. For example, in figure \ref{step2}, there is exactly one active component of $F_1$ in $U:=\{t_2,t_3\}$ as compared to two active components of $F_2$ in $U$. Even if the edges in the forests contain a feasible solution for the HSF, we do not terminate the algorithm if there is at least one active component $\overline{C}$ of $F_1$ such that $\sum_{S \subseteq \overline{C}} Y_1(S) < \sum_{S \subseteq \overline{C}} Y_2(S)$. For example, consider the snap shot of the algorithm in figure \ref{step6}. Each target in this snap shot is either connected to $d_1$ or $d_2$ and hence, one can possibly terminate the algorithm at this step. However, we find that the component $\overline{C}:=\{t_5,t_6,t_7,t_8\}$ of $F_1$ is still active and $\sum_{S \subseteq \overline{C}} Y_1(S) < \sum_{S \subseteq \overline{C}} Y_2(S)$, $i.e.$, the total price that $\overline{C}$ has paid till this iteration to get connected to $d_1$ is less than the total price that the $\overline{C}$ has already paid to get connected to $d_2$. Therefore, the algorithm continues to increase $Y_1(\overline{C})$ to check if all the vertices in $\overline{C}$ can get connected to $d_1$ at a lower cost. This feature is useful from the point of obtaining a good approximation ratio because the cost of the edges in the HSF has to be bounded in terms of the cost of the dual solution which in turn depends only on $\sum_{S\subseteq T}Y_1(S)$. Hence, the algorithm terminates only when all the components in $F_1$ become inactive. \\

  There are other possible ways of increasing the dual variables so that the primal-dual algorithm is simpler. For example, one can increase the dual variable associated with each active component in $F_1$ by the same amount while growing the dual variable of an active component in $F_2$ at a slower rate so that the associated constraints in (\ref{eq:dualbound}) are always tight in {\it each} iteration of the algorithm. Specifically, if a dual variable $Y_1(S)$ is increased by $\epsilon$, then the dual variable of each active component of $F_2$ in $S$ can be increased by $\frac{\epsilon}{k}$ where $k$ is the number of active components of $F_2$ in $S$. Even though this type of a dual increase will result in a simpler primal-dual procedure, the active components of $F_2$ could grow at different rates. If the active components grow at different rates in a forest, as pointed out in the case of the minimum spanning tree problem \cite{GoemansW95}, one can develop instances where the approximation ratio of the algorithm may be greater than 2.

  \begin{figure}[!t]
\centering
\includegraphics[scale=0.45]{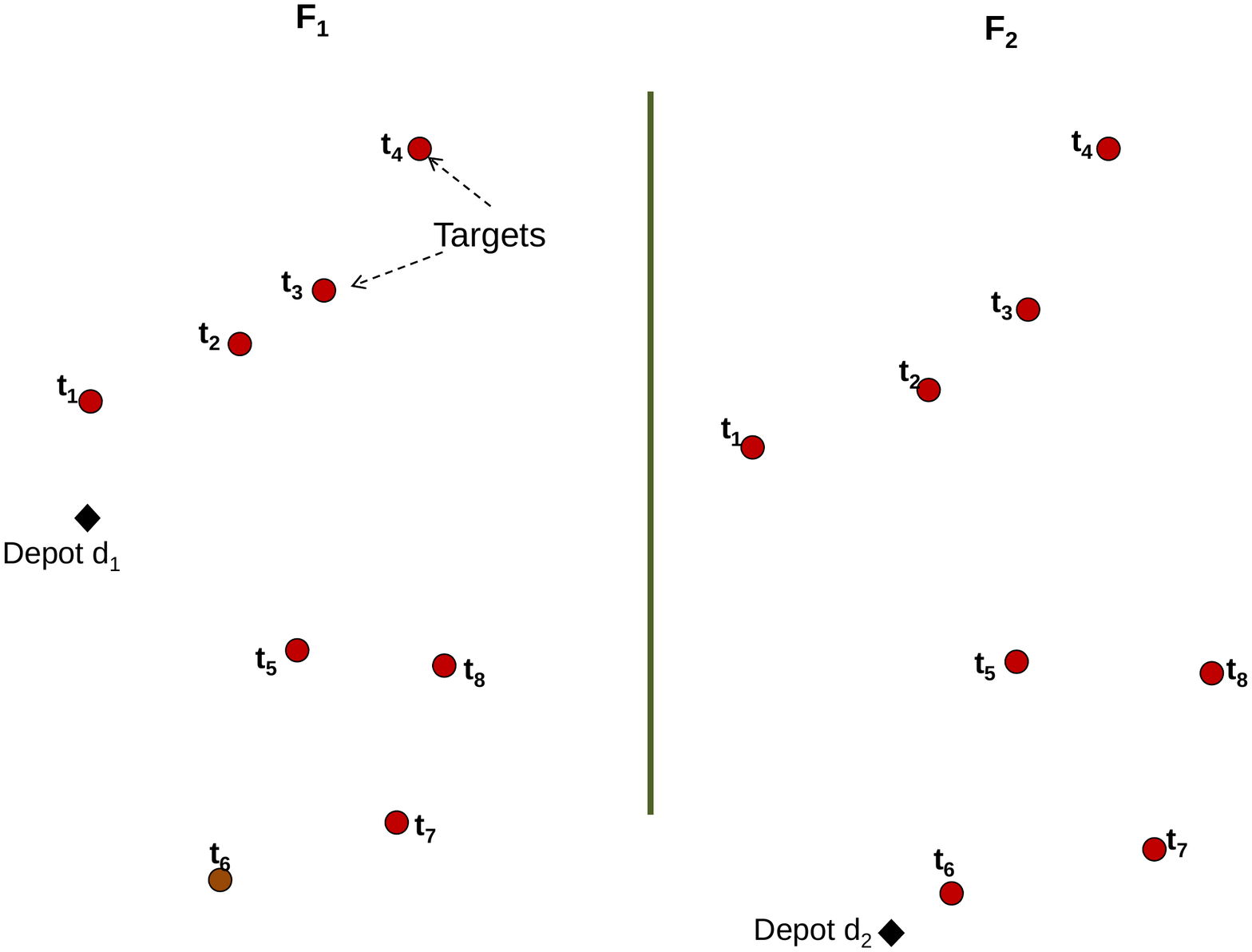}
\caption{An example illustrating the basic steps in the primal dual algorithm. There are
8 targets in this example. The forests $F_1$ and $F_2$ are initially empty. Each component that contains
a target is active. The components that contain the depots are inactive. }
\label{step1}
\end{figure}

\begin{figure}[!b]
\centering
\includegraphics[scale=0.45]{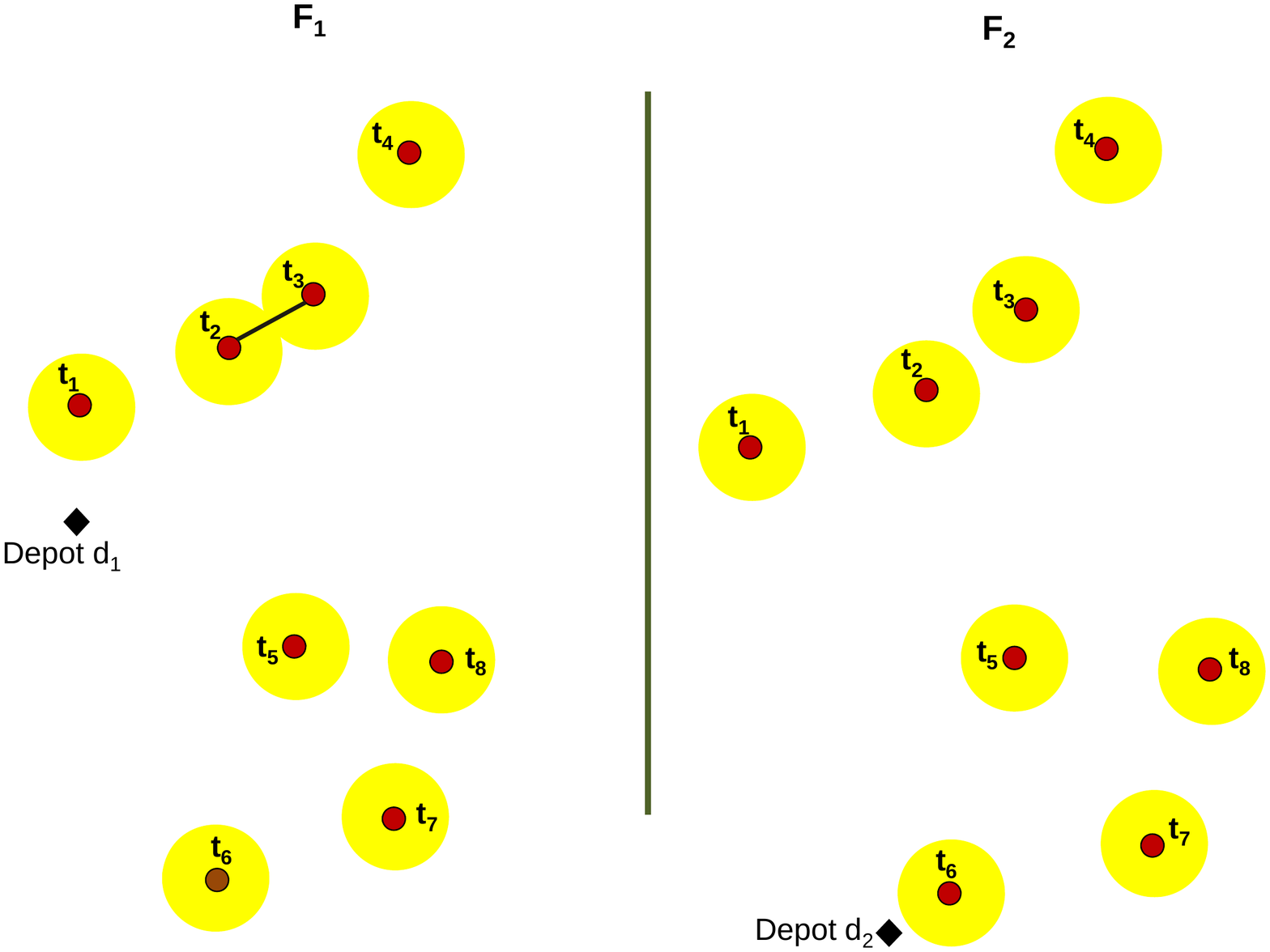}
\caption{Snap shot of the forests at the end of the first iteration. The radius of the circular region, $p_i(u):=\sum_{S: u \in S} Y_i(S) $, around a target $u$ in the forest $F_i$ is equal to the sum of the dual variables of all the components that contain $u$ in $F_i$. Edge $e:=(t_2,t_3)$ is added to $F_1$ as $p_1(t_2)+p_1(t_3)$  becomes equal to $cost^1_e$.
}
\label{step2}
\end{figure}

\begin{figure}[!t]
\centering
\includegraphics[scale=0.45]{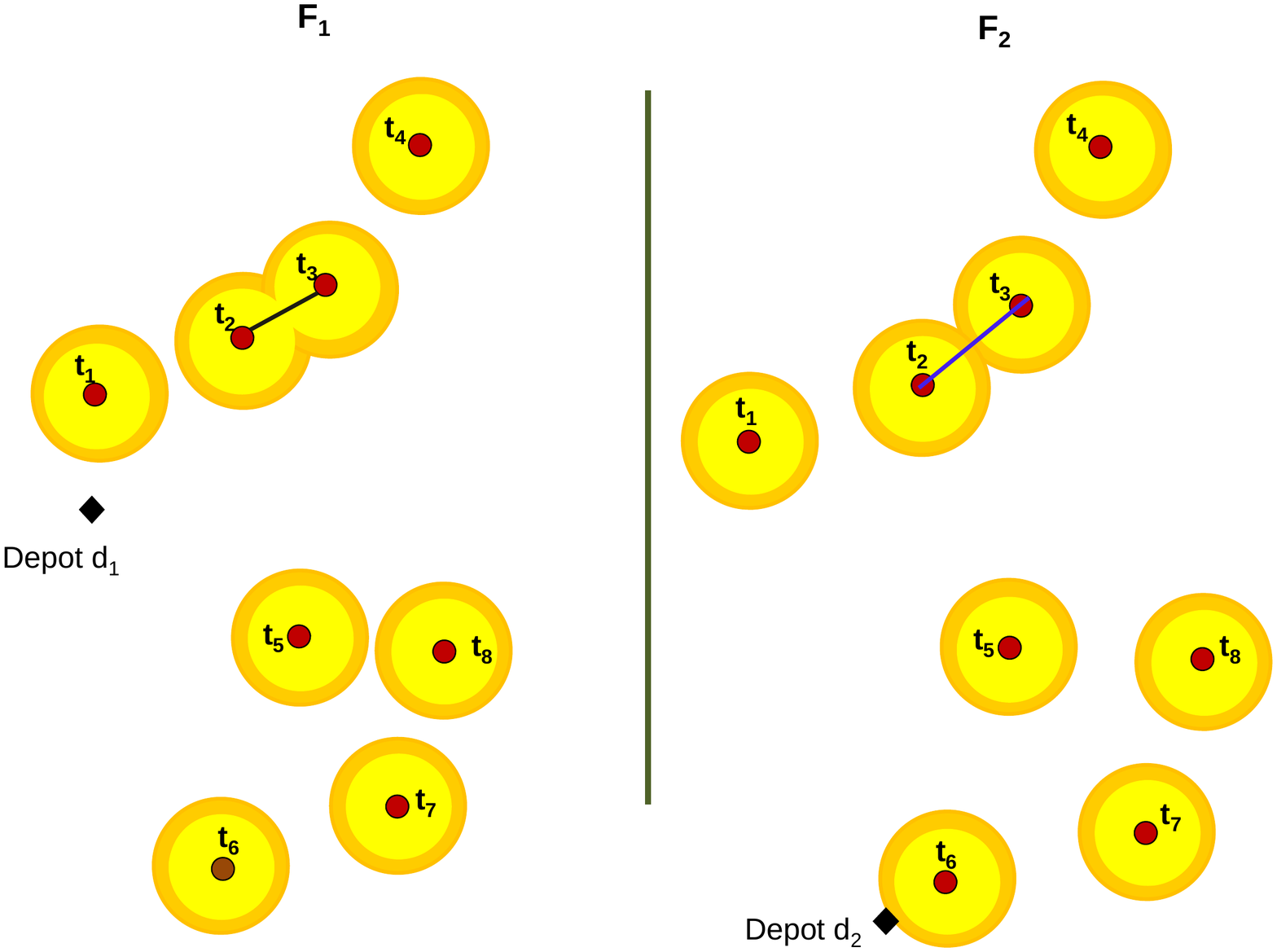}
\caption{Snap shot of the forests at the end of the second iteration. Edge ($t_2,t_3$) is added to $F_2$ as the sum of the prices paid by the components containing targets $t_2$ and $t_3$ becomes equal to the cost of constructing the edge ($t_2,t_3$) for the second vehicle.}
\label{step3}
\end{figure}

\begin{figure}[!b]
\centering
\includegraphics[scale=0.45]{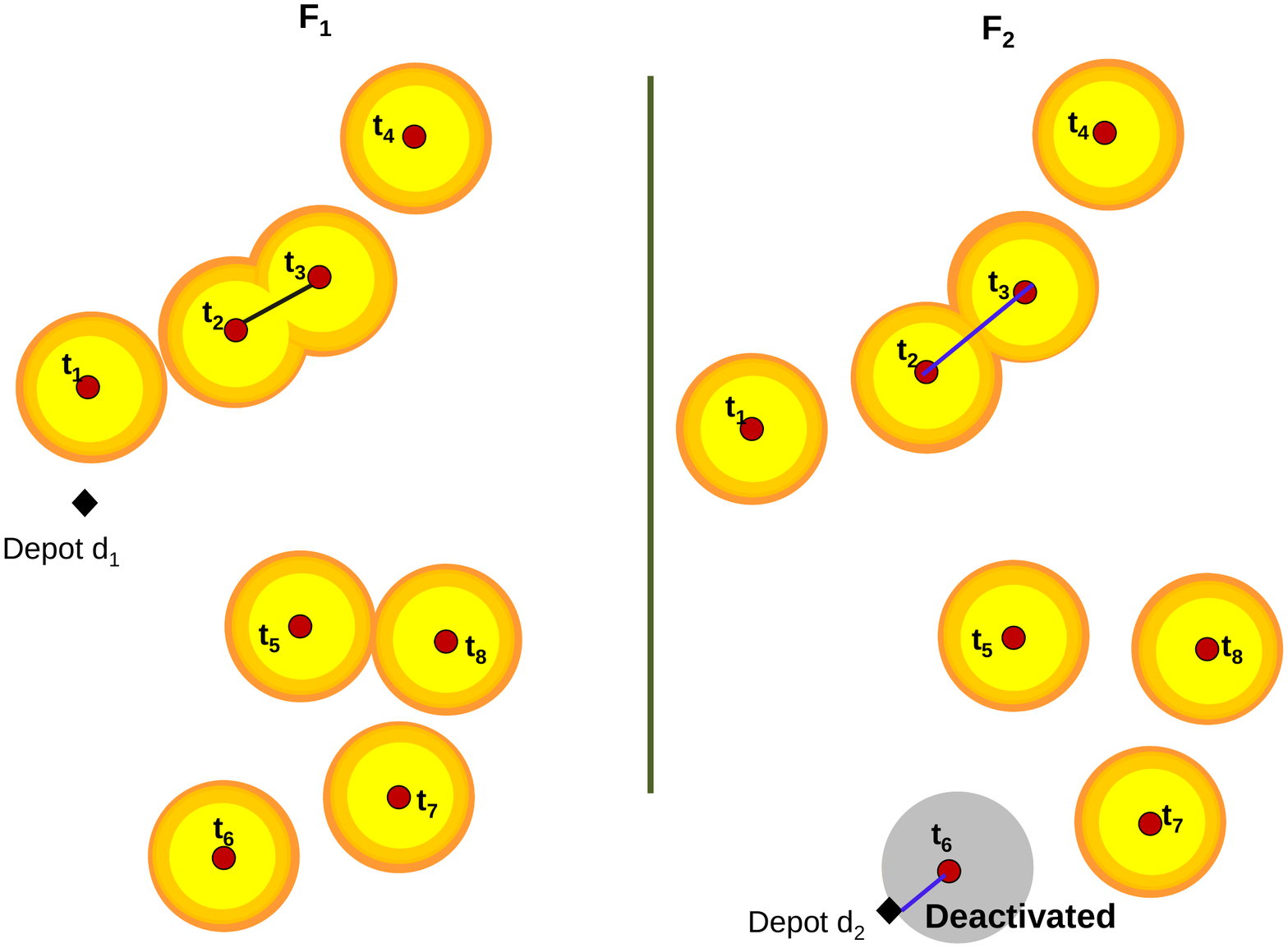}
\caption{Snap shot of the forests at the end of the third iteration. The constraint corresponding to the edge joining target $t_6$ and depot $d_2$ becomes tight. Edge ($t_6$,$d_2$) is added to $F_2$ and the merged component is deactivated as $t_6$ is now connected to $d_2$ in $F_2$. The dual variable $Y_2(\{t_6\})$ does not increase further and will serve as an upper bound on $Y_1(\{t_6\})$.}
\label{step4}
\end{figure}

\begin{figure}[!t]
\centering
\includegraphics[scale=0.44]{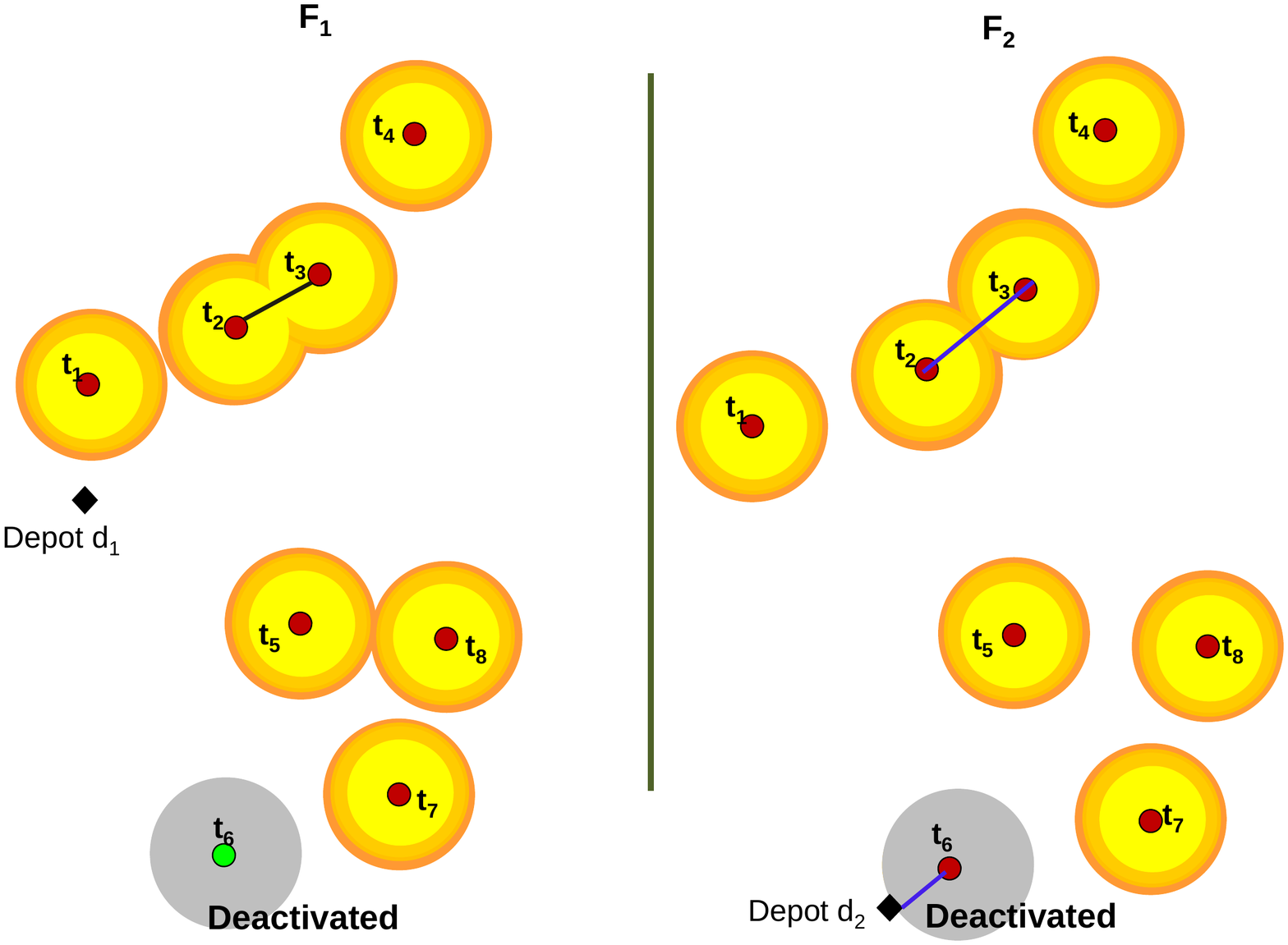}
\caption{Snap shot of the forests at the end of the fourth iteration. Component \{$t_6$\} in $F_1$ is deactivated because $Y_1(\{t_6\})$ becomes equal to $Y_2(\{t_6\})$. }
\label{step5}
\end{figure}

\begin{figure}[!b]
\centering
\includegraphics[scale=0.44]{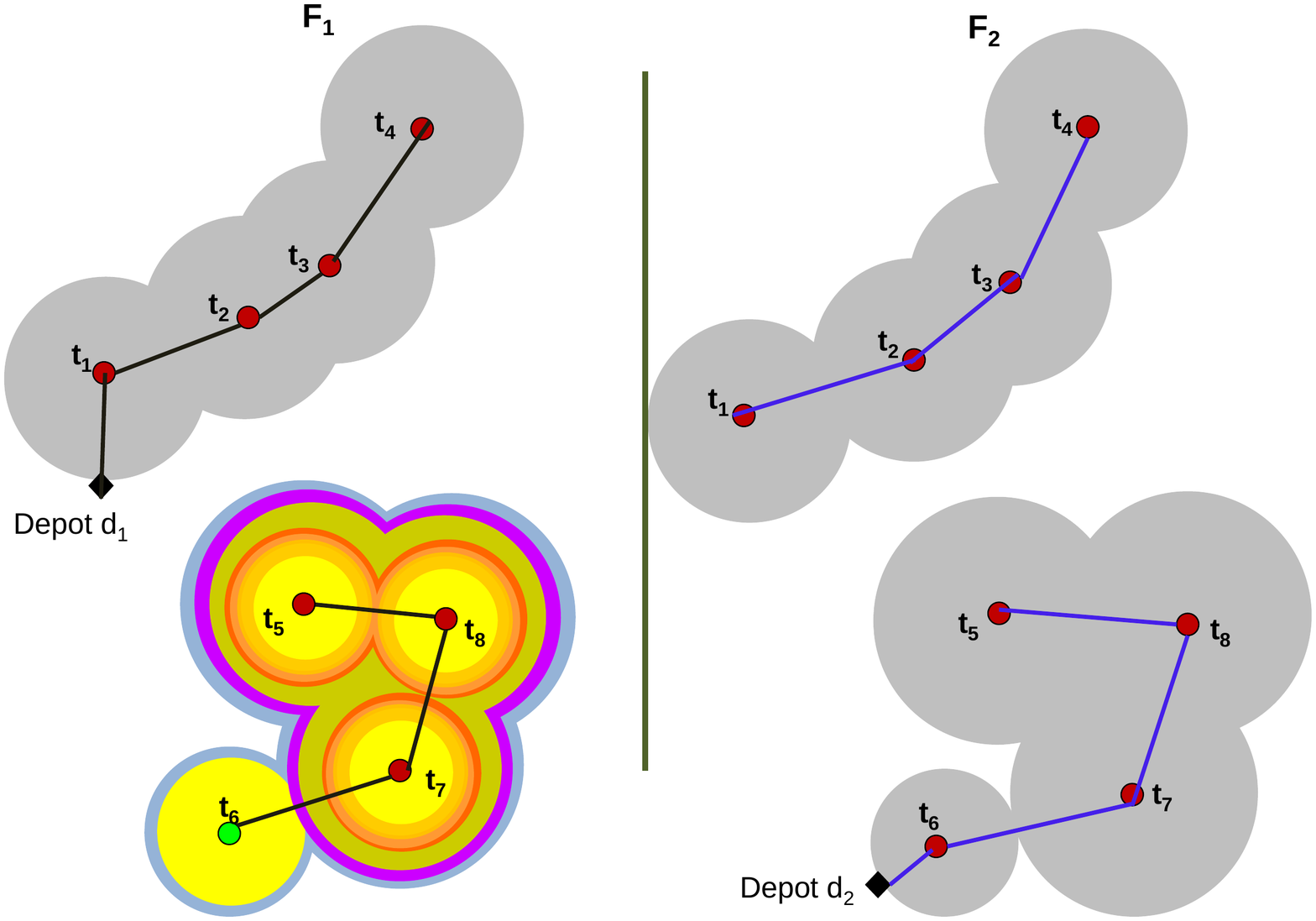}
\caption{Snap shot of the forests after few iterations of the algorithm. All the components are inactive except $\overline{C}:=\{t_5,t_6,t_7,t_8\}$ of $F_1$. Notice that all the targets are connected to one of the two depots. So, the algorithm can possibly stop if needed. However, it turns out that the total price paid by the components in $\overline{C}$ to get connected to the first depot is less than the total price the components in $\overline{C}$ have already paid for $F_2$. Therefore, $Y_1(\overline{C})$ is increased further in the next iteration to check if $\overline{C}$ can get connected to $d_1$ at a lower cost. }
\label{step6}
\end{figure}

\begin{figure}[!t]
\centering
\includegraphics[scale=0.47]{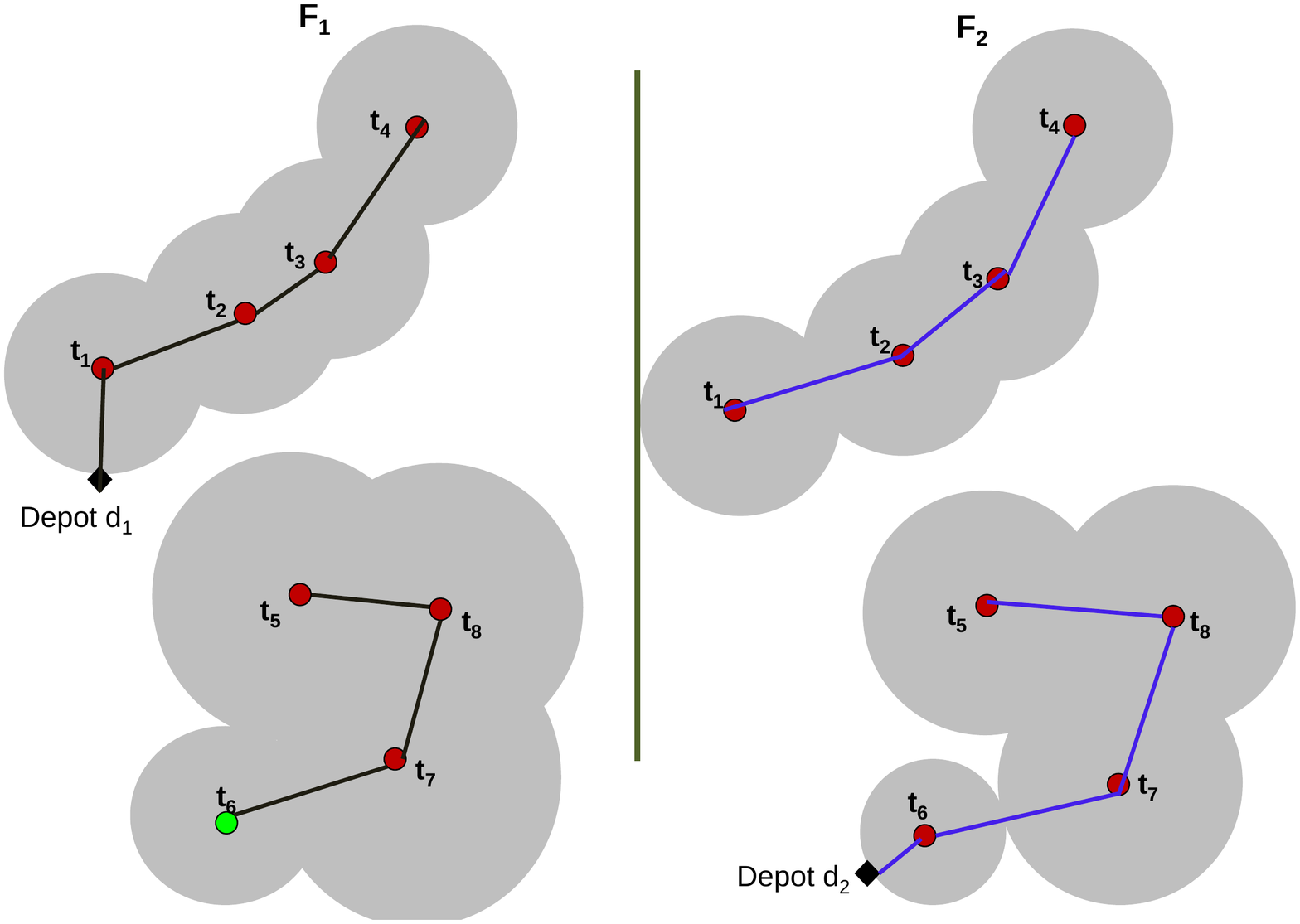}
\caption{Snap shot of the forests at the end of the main loop of the algorithm. $\overline{C}:=\{t_5,t_6,t_7,t_8\}$ of $F_1$ is deactivated because $\sum_{S \subseteq \overline{C}} Y_1(S)$ becomes equal to $\sum_{S \subseteq \overline{C}} Y_2(S)$. The main part of the algorithm terminates because all the components are now inactive.}
\label{step7}
\end{figure}

\begin{figure}[!b]
\centering
\includegraphics[scale=0.47]{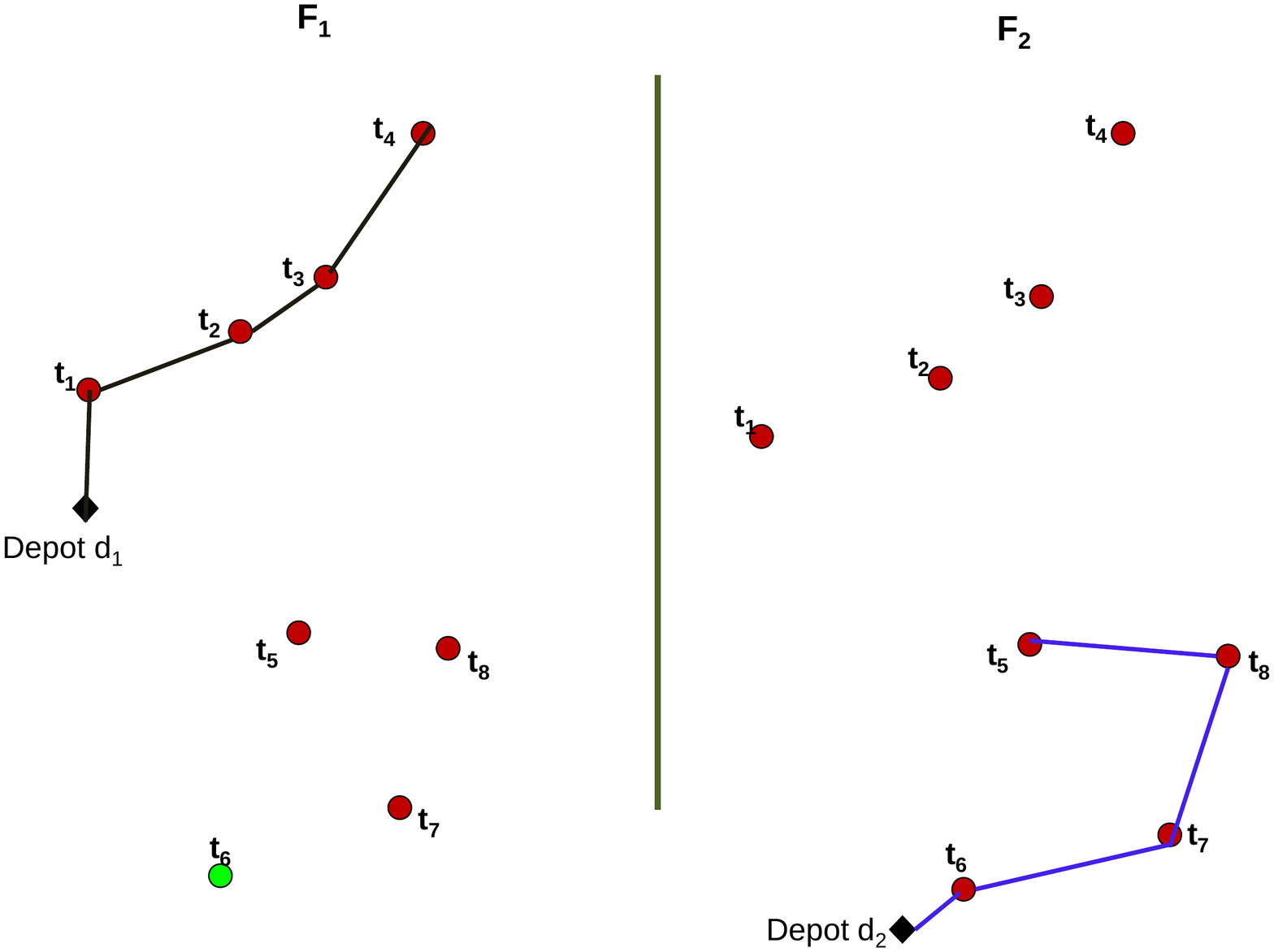}
\caption{The final output (HSF) of the primal-dual algorithm after the unnecessary edges are removed in the pruning step.}
\label{step8}
\end{figure}

\newpage

\section{Implementation details of the Primal-Dual Algorithm}

The initialization, the main steps and the final pruning step of the primal-dual algorithm are presented in {\bf Algorithms \ref{alg:initialization}} , {\bf \ref{alg:main} and {\bf \ref{alg:pruning}}}. For any $\forall C\in {\mathcal{C}}_1$, the internal variable $w(C)$ keeps track of $\sum_{S\subseteq C} Y_1(S)$, $i.e$, $w(C)=\sum_{S\subseteq C} Y_1(S)$. Similarly, $\forall C\in {\mathcal{C}}_1$, $Bound(C)$ keeps track of $\sum_{S\subseteq C} Y_2(S)$. Essentially, $w(C)$ and $Bound(C)$ are used to enforce the constraints in (\ref{eq:dualbound}). Initially, all the dual variables, $w(C)$ and $Bound(C)$ are set to zero. ({\it Refer to the initialization steps in algorithm \ref{alg:initialization}}). Also, each vertex in $V_1$ is initially unmarked. \\

As the components in ${\mathcal{C}}_1$ tend to merge first, we refer to the components in ${\mathcal{C}}_1$ as parents and the components in ${\mathcal{C}}_2$ as their children. For components $C_1\in {\mathcal{C}}_1$ and $C_2\in {\mathcal{C}}_2$, we define $C_1$ as the parent of $C_2$ and $C_2$ as a child of $C_1$ if $C_2\subseteq C_1$ and $d_2\notin C_2$. For any component $C_1 \in  {\mathcal{C}}_1$, we use $Children(C_1)$ to denote all the children of $C_1$ present in  ${\mathcal{C}}_2$. For any component $C_2 \in  {\mathcal{C}}_2, d_2\notin C_2$, we use $Parent(C_2)$ to denote the parent of $C_2$ present in ${\mathcal{C}}_1$. According to the definition, if $C_2$ contains the depot $d_2$, $C_2$ doesn't have a parent; however, to simplify the presentation, we let $Parent(C_2)$ be an empty set if $C_2$ contains $d_2$. At the start of the algorithm, for any target $v\in T$, $Children(\{v\})$ is assigned to be equal to $\{v\}$ and $Parent(\{v\})$ is assigned to be equal to $\{v\}$. Also, the components that consist of just the depots neither have a parent or a child ({\it Refer to the initialization steps in algorithm \ref{alg:initialization}}). \\

In each iteration of the algorithm, the dual variable corresponding to each of the \textit{active} components in ${\mathcal{C}}_1$ and  ${\mathcal{C}}_2$ are increased as much as possible by the \textit{same} amount until one of the constraints stated in (\ref{eq:dualfirst}-\ref{eq:dualbound}) becomes tight ({\it Refer to lines 2-5 of the algorithm \ref{alg:main}}). For any two disjoint components $C_{1x},C_{1y}\in {\mathcal{C}}_1$, consider the constraint in (\ref{eq:dualfirst}) corresponding to the edge $e=\{u,v\}$ that could potentially connect vertex $u$ in $C_{1x}$ to vertex $v$ in $C_{1y}$: $\sum_{S: e\in \delta_1(S)} Y_1(S) \leq cost^1_e$. Since $e$ has not yet been added to $F_1$, this constraint can be re-written as $\sum_{S: u\in S} Y_1(S) + \sum_{S: v \in S} Y_1(S)\leq cost^1_e$, or as $p_1(u) + p_1(v) \leq cost^1_e$. Therefore, to add an edge $(u,v)$ during the iteration, each of the dual variables of the active components have to be increased by an amount given by $\frac{cost^1_e-p_1(u)-p_1(v)}{active_1(C_{1x})+ active_1{(C_{1y})}}$ in order to make the constraint, $p_1(u) + p_1(v) \leq cost^1_e$, tight. Hence, in step 2 of the algorithm \ref{alg:main}, we basically find the minimum amount by which each of the dual variables of the active components in ${\mathcal{C}}_1$ have to be increased so that none of the constraints are violated and at least one of the constraints in (\ref{eq:dualfirst}) just becomes tight. Similarly, in step 3 of the algorithm \ref{alg:main}, we find the minimum amount by which each of the dual variables of the active components in ${\mathcal{C}}_2$ have to be increased so that none of the constraints are violated and at least one of the constraints in (\ref{eq:dualsecond}) just becomes tight. For $i=1,2$, note that $p_i(u)$ is increased during an iteration only if $u$ belongs to a component in ${\mathcal{C}}_i$ that is active; else $p_i(u)$ does not change.\\

If a constraint in (\ref{eq:dualfirst}) becomes tight for some edge $e \in E_1$, $F_1$ is augmented with this new edge and the two components (say $C_{1x},C_{1y}$ in $C_1$) connected by $e$ are merged to form a single connected component. The children of each of the two components $C_{1x},C_{1y}$ now together become the children of the resulting component $C_{1x} \bigcup C_{1y}$. The resulting component becomes inactive if it contains the depot $d_1$; otherwise, it is active. In the case when the resulting component becomes inactive, all the children of the resulting component also become inactive ({\it Refer to lines 17-27 of the algorithm \ref{alg:main}}). \\

Similarly, if one of constraints in (\ref{eq:dualsecond}) becomes tight for some edge $e \in E_2$, $F_2$ is augmented with this new edge and the two components (say $C_{2x},C_{2y}$ in $C_2$) connected by $e$ are merged to form a single connected component ({\it Refer to lines 29-39 of the algorithm \ref{alg:main}}). The resulting component becomes inactive if it contains the depot $d_2$; otherwise, it is active. In the case when the resulting component is active, the parent of either $C_{2x}$ or $C_{2y}$ is assigned as the parent of the resulting component (It turns out that due to our assumptions on the costs, when the algorithm enters this part of the implementation, both $C_{2x}$ and $C_{2y}$ must be active and must be the children of the same parent; we will show this result later in lemma \ref{lemma:activesets}).  In the case when the resulting component becomes inactive, and say $C_{2x}$ was the active component during the iteration which \textit{did not} contain the depot, the parent of $C_{2x}$ loses $C_{2x}$ as its child. \\ 

Once an active parent $\overline{C}$ loses all its children, $Bound(\overline{C})$ specifies the maximum value that can be attained by $w(\overline{C})$. Suppose an active component $\overline{C}\in {\mathcal{C}}_1$ does not have any children and the increase in the dual variables results in the constraint $w(\overline{C})\leq Bound(\overline{C})$ becoming tight. Then, the algorithm deactivates $\overline{C}$ and marks each of the unmarked vertices in the component with $\overline{C}$ ({\it Refer to lines 41-42 of the algorithm \ref{alg:main}}). \\

 The algorithm terminates when all the components in ${\mathcal{C}}_1$ become inactive. After termination, the algorithm makes one final pass at all the edges ({\it refer to algorithm \ref{alg:pruning}}) and removes any edge that is not required to be in the HSF. Basically, during the final step of the primal dual algorithm, any unnecessary edges in $F_1$ and $F_2$ are pruned further to find a tree for each of the vehicles. Specifically, the tree $F'_1$ corresponding to the first vehicle  is obtained from $F_1$ by removing as many edges as possible from $F_1$ so that the following properties hold: 1) All the unmarked vertices of $V_1$ are connected to the first depot $d_1$; 2) If any vertex with label $C$ is connected to the depot $d_1$, then any other vertex with a label $C'\supseteq C$ is also connected to the depot $d_1$. The tree $F'_2$ corresponding to the second vehicle is obtained from $F_2$ by removing as many edges as possible from $F_2$ such that any target not spanned by $F'_1$ is connected to $d_2$ in $F'_2$. \\

  Since the sum of the number of components in ${\mathcal{C}}_1$, the number of active components in ${\mathcal{C}}_1$ and the number of components in ${\mathcal{C}}_2$ decreases at least by one during each iteration, the primal-dual algorithm must terminate after at most $3|T|+2$ iterations.  Using the techniques given in \cite{GoemansW95}, this primal-dual algorithm can be implemented in $|T|^2 \log{|T|}$ steps.\\

\renewcommand{\baselinestretch}{1}
\begin{algorithm}
\begin{algorithmic}[0]
\STATE $F_1 \leftarrow \emptyset$;~~ $F_2 \leftarrow \emptyset$; ~~${\mathcal{C}}_1 \leftarrow \{ \{v\}: v \in V_1\}$;~~ ${\mathcal{C}}_2 \leftarrow \{ \{v\}: v \in V_2\}$
\FOR{$v \in V_1$}
\STATE Unmark $v$;~~ $p_1(v) \leftarrow 0$; ~~ $w(\{v\}) \leftarrow 0$; ~~ $Bound(\{v\}) \leftarrow 0$
\STATE If $v=d_1$, then $Children(\{v\}) \leftarrow \emptyset$, else $Children(\{v\}) \leftarrow \{v\}$
\STATE If $v=d_1$, then $active_1(\{v\}) = 0$, else $active_1(\{v\}) = 1$
\ENDFOR
\FOR{$v \in V_2$}
\STATE $p_2(v) \leftarrow 0$
\STATE If $v=d_2$, then $Parent(\{v\}) \leftarrow \emptyset$, else $Parent(\{v\}) \leftarrow \{v\}$
\STATE If $v=d_2$, then $active_2(\{v\}) = 0$, else $active_2(\{v\}) = 1$
\ENDFOR
\end{algorithmic}
\caption{{\it Primal-dual algorithm: Initialization}}
\label{alg:initialization}
\end{algorithm}
\renewcommand{\baselinestretch}{1.5}

\renewcommand{\baselinestretch}{1}

\begin{algorithm}
\begin{algorithmic}[1]
{
\WHILE{$\exists C \in {\mathcal{C}}_1$ such that $active_1(C)=1$}
\STATE Find edge $e_1 = (i,j) \in E_1$ with $i\in C_{1x},j\in C_{1y}$ where $C_{1x},C_{1y} \in {\mathcal{C}}_1, C_{1x} \neq C_{1y}$ that minimizes
    $\varepsilon_1 = \frac{(cost^1_{e_1}-p_1(i)-p_1(j))}{active_1(C_{1x}) + active_1(C_{1y})}$ \label{alg_addedge1}
\STATE Find edge $e_2 = (i,j) \in E_2$ with $i\in C_{2x},j\in C_{2y}$ where $C_{2x},C_{2y} \in {\mathcal{C}}_2, C_{2x} \neq C_{2y}$ that minimizes
    $\varepsilon_2 = \frac{(cost^2_{e_2}-p_2(i)-p_2(j))}{active_2(C_{2x}) + active_2(C_{2y})}$
\STATE Let ${\mathcal{\mathfrak{C}}} := \{C: active_1(C)=1, Children(C)= \emptyset, C\in {\mathcal{C}}_1\}$. Find $\overline{C}\in {\mathfrak{{C}}}$ that minimizes $\varepsilon_3 = Bound(\overline{C})-w(\overline{C})$
    \STATE $\varepsilon_{min}=\min(\varepsilon_1,\varepsilon_2,\varepsilon_3)$
    \FOR{each active component $C \in {\mathcal{C}}_1$}
    \STATE $w(C) \leftarrow w(C) + \varepsilon_{min}$
    \STATE For all $v\in C$,  $p_1(v) \leftarrow p_1(v) + \varepsilon_{min}$
    \STATE $Bound(C) \leftarrow Bound(C) + \varepsilon_{min} |Children(C)|$
    \ENDFOR
    \FOR{each active component $C \in {\mathcal{C}}_2$}
    \STATE  For all $v\in C$,  $p_2(v) \leftarrow p_2(v) + \varepsilon_{min}$
    \ENDFOR

    \STATE \textbf{switch} $\varepsilon_{min}$
    \STATE //\verb"Comment: If more than one value in" $\{\varepsilon_1,$$\varepsilon_2,$$\varepsilon_3\}$ \verb"is equal" \verb"to" $\varepsilon_{min}$, \verb"then give priority" \verb"first to" {\it Case} $\varepsilon_1$, \verb"then to" {\it Case} $\varepsilon_2$ \verb"and finally to" {\it Case} $\varepsilon_3$ \vspace{0.3cm}
    \STATE \textit{Case} $\varepsilon_1$:
     \STATE  \hspace{0.3cm} $F_1 \leftarrow F_1 \bigcup \{e_1\}$ \label{alg:Case1Starts}
    \STATE  \hspace{0.3cm}  ${\mathcal{C}}_1 \leftarrow {\mathcal{C}}_1 \bigcup \{C_{1x}\bigcup C_{1y}\} - C_{1x} - C_{1y} $
     \STATE  \hspace{0.3cm}  $w(C_{1x}\bigcup C_{1y}) \leftarrow w(C_{1x}) + w(C_{1y})$
     \STATE  \hspace{0.3cm}  $Children(C_{1x}\bigcup C_{1y}) \leftarrow Children(C_{1x})\bigcup Children(C_{1y}) $
    \STATE \hspace{0.3cm}  For all $C \in Children(C_{1x}\bigcup C_{1y})$, $Parent(C) \leftarrow C_{1x}\bigcup C_{1y}$
    \STATE \hspace{0.3cm} $Bound(C_{1x}\bigcup C_{1y})\leftarrow Bound(C_{1x}) + Bound(C_{1y})$
    \STATE  \hspace{0.3cm}  \textbf{if }$d_1 \in C_{1x}\bigcup C_{1y}$,\textbf{ then}
    \STATE  \hspace{0.3cm}    $active_1(C_{1x}\bigcup C_{1y})=0$
    \STATE  \hspace{0.3cm}   $active_2(C) = 0$ for all $C \in Children(C_{1x}\bigcup C_{1y})$
   \STATE  \hspace{0.3cm}    \textbf{else }$active_1(C_{1x}\bigcup C_{1y})=1$
   \STATE  \hspace{0.3cm}     \textbf{end}

    \STATE \textit{Case} $\varepsilon_2$:
         \STATE  \hspace{0.3cm} $F_2 \leftarrow F_2 \bigcup \{e_2\}$
    \STATE  \hspace{0.3cm}  ${\mathcal{C}}_2 \leftarrow {\mathcal{C}}_2 \bigcup \{C_{2x}\bigcup C_{2y}\} - C_{2x} - C_{2y} $
        \STATE \hspace{0.3cm}\textbf{ if }{$d_2 \in C_{2x} \bigcup C_{2y}$ }\textbf{ then}
        \STATE \hspace{0.6cm} $active_2(C_{2x}\bigcup C_{2y})\leftarrow 0$
        \STATE \hspace{0.6cm} $Parent(C_{2x}\bigcup C_{2y}) \leftarrow \emptyset$
    \STATE \hspace{0.6cm} Let $C \in \{C_{2x},C_{2y}\}$ such that $d_2 \notin C$; $Children(Parent(C)) \leftarrow Children(Parent(C)) - C$
    \STATE \hspace{0.3cm} \textbf{else} $active_2(C_{2x}\bigcup C_{2y})\leftarrow 1$
               \STATE \hspace{0.6cm}   $C_{temp} \leftarrow Parent(C_{2x})$
    \STATE \hspace{0.6cm}  $Parent(C_{2x}\bigcup C_{2y}) \leftarrow C_{temp}$
            \STATE \hspace{0.6cm}  $Children(C_{temp})\leftarrow Children(C_{temp})\bigcup \{C_{2x}\bigcup C_{2y}\} - C_{2x} - C_{2y} $
    \STATE \hspace{0.3cm} \textbf{ end if}
    \STATE \textit{Case} $\varepsilon_3$:
    \STATE  \hspace{0.3cm} $active_1(\overline{C}) \leftarrow 0$
    \STATE \hspace{.3cm} Mark all the unlabeled vertices of $\overline{C}$ with label $\overline{C}$
    \STATE \textbf{end switch}

\ENDWHILE
}
\end{algorithmic}
\caption{: \textit{Primal-dual algorithm - Main steps}}
\label{alg:main}
\end{algorithm}

\begin{algorithm}
\begin{algorithmic}[1]
{
\STATE $F'_1$ is obtained from $F_1$ by removing as many edges as possible from $F_1$ so that the following properties hold: 1) All the unmarked vertices of $V_1$ are connected to the first depot $d_1$; 2) If any vertex with label $C$ is connected to the depot $d_1$, then any other vertex with a label $C'\supseteq C$ is also connected to the depot $d_1$.
\STATE $F'_2$ is obtained from $F_2$ by removing as many edges as possible from $F_2$ such that any target not spanned by $F'_1$ is spanned by $F'_2$.
}
\end{algorithmic}
\caption{: \textit{Primal-dual algorithm - Pruning step}}
\label{alg:pruning}
\end{algorithm}

\renewcommand{\baselinestretch}{1.5}

\newpage

 \subsection{Properties of the primal-dual algorithm}
 \vspace{0.3cm}

Consider any target $u\in T$. At the start of the $k^{th}$ iteration, let $C^k_1(u)$ denote the component in ${\mathcal{C}}_1$ containing $u$, and $C^k_2(u)$ represent the component in ${\mathcal{C}}_2$ containing $u$.\\

\begin{lemma}\label{lemma:activesets}
The following statements are true for all $k$:\newline
\begin{enumerate}
\item $C^k_2(u)$ is always a child of $C^k_1(u)$, $i.e.$, $C^k_2(u) \subseteq C^k_1(u)$ unless $C^k_2(u)$ contains the depot $d_2$.
\item $active_1(C^{k}_1(u))\geq active_2(C^{k}_2(u))$. \\
\end{enumerate}
 \end{lemma}

 \begin{proof}
 Let us prove this lemma by induction. At the start of the first iteration, $C^1_1(u)=C^1_2(u)=\{u\}$ and the components $C^1_1(u)$, $C^1_2(u)$ are both active. Therefore, lemma 1.1 and lemma 1.2 are correct for $k=1$. Now, let us assume that the statements in the lemma are true for the $l^{th}$ iteration for any $l=1,\cdots,k$. As $active_1(C^l_1(u))\geq active_2(C^l_2(u))$ for any $l=1\cdots,k$, it follows that $p_1(u)\geq p_2(u)$ at the start of the $k^{th}$ iteration.

\textit{Proof of lemma 1.1:}
During the $k^{th}$ iteration, there are three possible cases for the components $C^k_1(u)$ and $C^k_2(u)$: 1) $C^k_1(u)$ merges with another component in ${\mathcal{C}}_1$, or, 2) $C^k_2(u)$ merges with another component in ${\mathcal{C}}_2$, or, 3) $C^k_1(u)$ gets deactivated because its corresponding constraint in (\ref{eq:dualbound}) becomes tight. It is easy to note that $C^{k+1}_2(u)$ will remain a child of $C^{k+1}_1(u)$ in the first case. $C^k_1(u)$ can get deactivated as in the third case only when $C^k_1(u)$ does not have any children, $i.e.$, $C^k_2(u)$ already contains $d_2$. Therefore, lemma 1.1 is true by default in the third case.

Let us now examine the second case. If $C^k_2(u)$ is active and merges with a component that contains the depot $d_2$, then lemma 1.1 is true for $l=k+1$ by default. If $C^k_2(u)$ is active and merges with another active component $C^k_2(v)$ corresponding to target $v$, we claim that both $C^k_2(u)$ and $C^k_2(v)$ must have the same parent. If this is not true, note that  \begin{equation}\label{varepsilon}
\varepsilon_1= \frac{cost^1_{(u,v)}-p_1(u)-p_1(v)}{active_1(C^k_1(u)) + active_1(C^k_1(v))} \leq  \frac{cost^2_{(u,v)}-p_2(u)-p_2(v)}{active_2(C^k_2(u)) + active_2(C^k_2(v))} =\varepsilon_2.
\end{equation}
Therefore, the algorithm \ref{alg:main} will not merge $C^k_2(u)$ and $C^k_2(v)$ unless it merges the parents of $C^k_2(u)$ and $C^k_2(v)$. If $C^k_2(u)$ and $C^k_2(v)$ have the same parent, it then follows that the merged component $C^{k+1}_2(u)$ will be a child of $C^{k+1}_1(u)$.

If $C^k_2(u)$ is inactive because its parent contains $d_1$, we claim that $C^k_2(u)$ will never merge with any other component. If this claim is not true and say $C^k_2(u)$ (which is inactive) merges with some other component $C^k_2(v)$ corresponding to target $v$, then $C^k_1(u) \neq C^k_1(v)$ and $C^k_2(v)$ must be active. Again from equation (\ref{varepsilon}), the algorithm will prefer to merge $C^k_1(u)$ and $C^k_1(v)$ before merging their children, $i.e.$, $C^k_2(u)$ and $C^k_2(v)$. But, once $C^k_1(u)$ and $C^k_1(v)$ are merged, the component $C^k_2(v)$ becomes a child of $C^k_1(u)\bigcup C^k_1(v)$ and as a result will be deactivated. Therefore, $C^k_2(u)$ will remain inactive and will never merge with any other component during the $k^{th}$ iteration. Hence, lemma 1.1 is true by default.

\textit{Proof of lemma 1.2:}

 If {\bf $C^k_2(u)$ is inactive}, either $C^k_2(u)$ must contain the depot $d_2$ or its parent $C^k_1(u)$ must contain the depot $d_1$.
 \begin{itemize}
 \item If $C^k_2(u)$ already contains $d_2$, then $C^{k+1}_2(u)$ must also be inactive. Therefore, $active_1(C^{k+1}_1(u))\geq active_2(C^{k+1}_2(u))=0$.
 \item If $C^k_2(u)$ is inactive because its parent $C^k_1(u)$ contains $d_1$, then we have already shown in lemma 1.1 that $C^k_2(u)$ can never merge with any other component during the $k^{th}$ iteration. Therefore, $active_1(C^{k+1}_1(u))\geq active_2(C^{k+1}_2(u))$.\\
 \end{itemize}

 If \textbf{${C^k_2(u)}$ is active}, then $active_1(C^k_1(u))\geq active_2(C^k_2(u))$ implies that $C^k_1(u)$ is also active. From lemma 1.1 it follows that $C^k_1(u)$ is a parent of $C^k_2(u)$. Since the component, $C^k_1(u)$, has at least one active child in $C^k_2(u)$, $C^k_1(u)$ can never become inactive due to its associated constraint in (\ref{eq:dualbound}) during the $k^{th}$ iteration. The only way $C^{k}_1(u)$ can lead to an inactive $C^{k+1}_1(u)$ is if $C^k_1(u)$ merges with another component containing $d_1$ during the iteration in which case all the children of $C^k_1(u)$ including $C^k_2(u)$ also get deactivated. Therefore, $active_1(C^{k+1}_1(u))\geq active_2(C^{k+1}_2(u))$.

 \end{proof}

Let $\mathfrak{X}$ denote the set of vertices not spanned by $F'_1$. Based on the label of each vertex in $\mathfrak{X}$, $\mathfrak{X}$ can be partitioned into disjoint, deactivated components $\overline{C}_1, \overline{C}_2,\cdots, \overline{C}_m $ where each $\overline{C}_i$ denotes the maximal label of its respective component. The following lemma shows that the primal-dual algorithm produces a feasible solution in which each target is connected to exactly one depot.

\begin{lemma}\label{lemma:feasibility}
  The algorithm produces a feasible, heterogeneous spanning forest, $i.e.$, the trees specified by the collection of edges in $F'_1$ and $F'_2$ connect each of the targets to one of the depots. Any vertex spanned by the edges in $F'_1$ is not spanned by the edges in $F'_2$ and vice versa.
\end{lemma}
\begin{proof}
The algorithm terminates when all the sets of ${\mathcal{C}}_{1}$ become inactive. This is only possible if each of the targets in $T$ is either connected to $d_1$ or $d_2$. Note that $F'_1$ is formed from $F_1$ such that each of the unmarked vertices remain connected to $d_1$. The only vertices not spanned by $F'_1$ are some of the marked vertices. These vertices were marked because the components in
${\mathcal{C}}_1$ that span these vertices were deactivated for making their associated constraints in (\ref{eq:dualbound}) tight. In addition, a component in ${\mathcal{C}}_1$ can become deactivated due to a constraint in (\ref{eq:dualbound}) only if it has already lost all its children, $i.e.$, each of these vertices in the component is already connected to $d_2$. Therefore, by the construction of $F'_2$, each of the marked vertices not spanned by $F'_1$ must be connected to $d_2$ and spanned by $F'_2$. Hence, the algorithm produces a feasible, heterogeneous spanning forest.

Consider any deactivated component $\overline{C}_i \subseteq \mathfrak{X}$. $\overline{C}_i$ can get deactivated during an iteration only if $\overline{C}_i$ does not have children and $\sum_{S\subseteq \overline{C}_i} Y_1(S) = w(\overline{C}_i)= Bound(\overline{C}_i) = \sum_{S \subseteq \overline{C}_i} Y_2(S)$. Note that $\overline{C}_i$ could have lost all its children only if all the targets in $\overline{C}_i$ are already connected to $d_2$ in $F_2$. Also, {\it during the iteration} when $\overline{C}_i$ gets deactivated, no target $u \in \overline{C}_i$ is connected to any other target $v \in T\setminus \overline{C}_i$ in $F_1$.
As a result, from lemma \ref{lemma:activesets}, we claim that $u$ does not have an adjacent vertex $v$ in $F_2$ such that $v\in T\setminus \overline{C}_i$. If this claim is not true, then from lemma \ref{lemma:activesets} and equation (\ref{varepsilon}), it follows that the algorithm would have added edge $(u,v)$ to $F_1$ before adding $(u,v)$ to $F_2$. Since target $u$ is not connected to target $v\in T\setminus \overline{C}_i$ in $F_1$, $u$ and $v$ cannot be connected in $F_2$. Therefore, during the construction of $F'_2$, all the edges that are incident on any vertex $u \notin \mathfrak{X}$ can be dropped. Hence, any vertex spanned by the edges in $F'_1$ is not spanned by the edges in $F'_2$ and vice versa.
\end{proof}

The main result of this article is in the following subsection.

\subsection{Proof of the Approximation Ratio}

\begin{theorem}\label{theorem:2}
The primal-dual algorithm produces a tree with edges denoted by $F'_1$ for the first vehicle and a tree with edges denoted by $F'_2$ for the second vehicle such that the cost of the edges in these trees is bounded by the cost for the dual problem, $i.e.$,

{{ \[\sum_{e\in F'_1} cost^1_e + \sum_{e\in F'_2} cost^2_e \leq 2\sum_{S \subseteq T}Y_1(S). \]
}}
Since $2\sum_{S \subseteq T}Y_1(S)$ is a lower bound to the optimal cost of the 2DHTSP, it follows that the cost of the HSF found by the primal dual algorithm is at most equal to the optimal cost of the 2DHTSP. This provides a 2-approximation algorithm for the 2DHTSP.
\end{theorem}

\begin{proof} In order to prove the above theorem, we first simplify the dual cost obtained by the algorithm as follows:

{\begin{align}
2\sum_{S \subseteq T}Y_1(S) 
& = 2 \sum_{S \subseteq T,\newline S \nsubseteq \overline{C}_i,i=1,..,m}Y_1(S) + 2\sum_{i=1}^m\sum_{S \subseteq \overline{C}_i}Y_1(S) \nonumber \\
& = 2\sum_{S \subseteq T,S \nsubseteq \overline{C}_i,i=1,..,m}Y_1(S) + 2\sum_{i=1}^m\sum_{S \subseteq \overline{C}_i}Y_2(S) \label{cost_dual}.
\end{align}
}
Now, we express the cost of the edges in the first tree in terms of the dual variables as follows. Note that edge $e$ is added to $F_1$ and
consequently appears in $F'_1$ only if the corresponding constraint in (\ref{eq:dualfirst}) is tight, $i.e.$, $cost^1_e = \sum_{S:e \in \delta_1(S)} Y_1(S)$. Therefore,

{ \begin{align*}
\sum_{e\in F'_1} cost^1_e& = \sum_{e\in F'_1} \sum_{S:e \in \delta_1(S)} Y_1(S) \\
& = \sum_{S \subseteq T } Y_1(S) |F'_1 \bigcap \delta_1(S)|.
\end{align*}
}
Since $F'_1 \bigcap \delta_1(S)=0$ for any $S\subseteq \overline{C}_i$, we can further simplify the above equation to
{ \begin{align}
\sum_{e\in F'_1} cost^1_e& = \sum_{S \subseteq T, S \nsubseteq \overline{C}_i,i=1,..,m } Y_1(S) |F'_1 \bigcap \delta_1(S)| \label{cost_primal_1}.
\end{align}
}
Similarly, we can also express the cost of the edges in the second tree in terms of the dual variables as follows. From lemma \ref{lemma:feasibility}, note that $F'_2$ can be decomposed into a set of disjoint sets $F'_{2i}$ where each $F'_{2i}$ consists of edges that form a tree spanning each target from $\overline{C}_{i}$ and the depot $d_2$. An edge $e$ is added to $F_{2}$ and consequently appears in $F'_{2i}$ only if the corresponding constraint in (\ref{eq:dualsecond}) is tight, $cost^2_e = \sum_{S:e \in \overline{\delta}_{2i}(S), S\subseteq \overline{C}_i} Y_2(S)$ where $\overline{\delta}_{2i}(S)$ consists of all the edges with one endpoint in $S$ and another end point in $\overline{C}_i\bigcup \{d_2\}\setminus S$.

{
\begin{align}
\sum_{e\in F'_2}cost^2_e & = \sum_{i=1}^m \sum_{e \in F'_{2i}} cost^2_e \nonumber \\
 & = \sum_{i=1}^m \sum_{e\in F'_{2i}} \sum_{S:e \in \overline{\delta}_{2i}(S), S\subseteq \overline{C}_i} Y_2(S) \nonumber \\
 & = \sum_{i=1}^m  \sum_{S \subseteq \overline{C}_i} Y_2(S) ~|F'_{2i} \bigcap \overline{\delta}_{2i}(S)|. \label{cost_primal_2}
\end{align}
}

Therefore, from equations (\ref{cost_dual}), (\ref{cost_primal_1}), (\ref{cost_primal_2}), the proof for the theorem reduces to showing the following result:
{
\begin{align}
\sum_{S \subseteq T, S \nsubseteq \overline{C}_i,i=1,..,m } Y_1(S) |F'_1 \bigcap \delta_1(S)| + \sum_{i=1}^m  \sum_{S \subseteq \overline{C}_i} Y_2(S) ~|F'_{2i} \bigcap \overline{\delta}_{2i}(S)| \\ \leq 2\sum_{S \subseteq T,S \nsubseteq \overline{C}_i,i=1,..,m}Y_1(S) + 2\sum_{i=1}^m\sum_{S \subseteq \overline{C}_i}Y_2(S). \end{align}
}

The above result can be shown by proving that during any iteration, the increase in the primal cost (the left-hand side of the above inequality) is at most equal to the increase in the dual cost (the right-hand side of the above inequality). To see this, let us choose any iteration of the primal-dual algorithm. At the start of this iteration, let $N_{a}$ be the set of all the active components in ${\mathcal{C}}_1$ such that each active component in this set is not a subset of $\mathfrak{X}$ and $N_{d}$ be the set of all the inactive components in ${\mathcal{C}}_1$ such that each inactive component in this set is not a subset of $\mathfrak{X}$. Note that one of inactive components of $N_d$ must consist of the depot $d_1.$ For $i=1,\cdots,m$, let $M_{ai}$ denote the set of all the active components in ${\mathcal{C}}_2$ such that each active component in this set is a subset of $\overline{C}_i$. Also, let $M_{d}$ denote the inactive component in ${\mathcal{C}}_2$ that consists of the depot $d_2$.

Now, form a graph $H_1$ with components in $N_a\bigcup N_d$ as its vertices and edges $e\in F'_1 \bigcap  \delta_1(C)$ for $C\in N_a\bigcup N_d$ as edges of $H_1$. $H_1$ is a tree that spans all the vertices in $N_a\bigcup N_d$. Similarly, form a graph $H_{2i}$ with components in $M_{ai}\bigcup M_d$ as its vertices and edges $e\in F'_{2i} \bigcap {\delta}_2(C) $ for $C\in M_{ai}\bigcup \{M_d\}$ as edges of $H_{2i}$. $H_{2i}$ is a tree that spans all the vertices in $M_{ai}\bigcup \{M_d\}$.

Let $deg(v,G)$ represent the degree of vertex $v$ in graph $G$. During the iteration, the dual variable corresponding to each of the active components is increased by $\varepsilon_{min}$. As the result, the left hand side of the inequality will increase by $\varepsilon_{min}(\sum_{v\in N_a} deg(v,H_1) + \sum_{i=1}^m  \sum_{v\in M_{ai}} deg(v,H_{2i}))$ whereas the right hand side of the inequality will increase by $2\varepsilon_{min} (N_a + \sum_{i=1}^m M_{ai})$. Therefore, basically, the proof is complete if we can show that

{{
\begin{equation}
\sum_{v\in N_a} deg(v,H_1) + \sum_{i=1}^m  \sum_{v\in M_{ai}} deg(v,H_{2i}) \leq 2 (|N_a| + \sum_{i=1}^m |M_{ai}|).
\end{equation}
}}

\begin{figure}[htbp]
 \centering
  \includegraphics[scale=0.5]{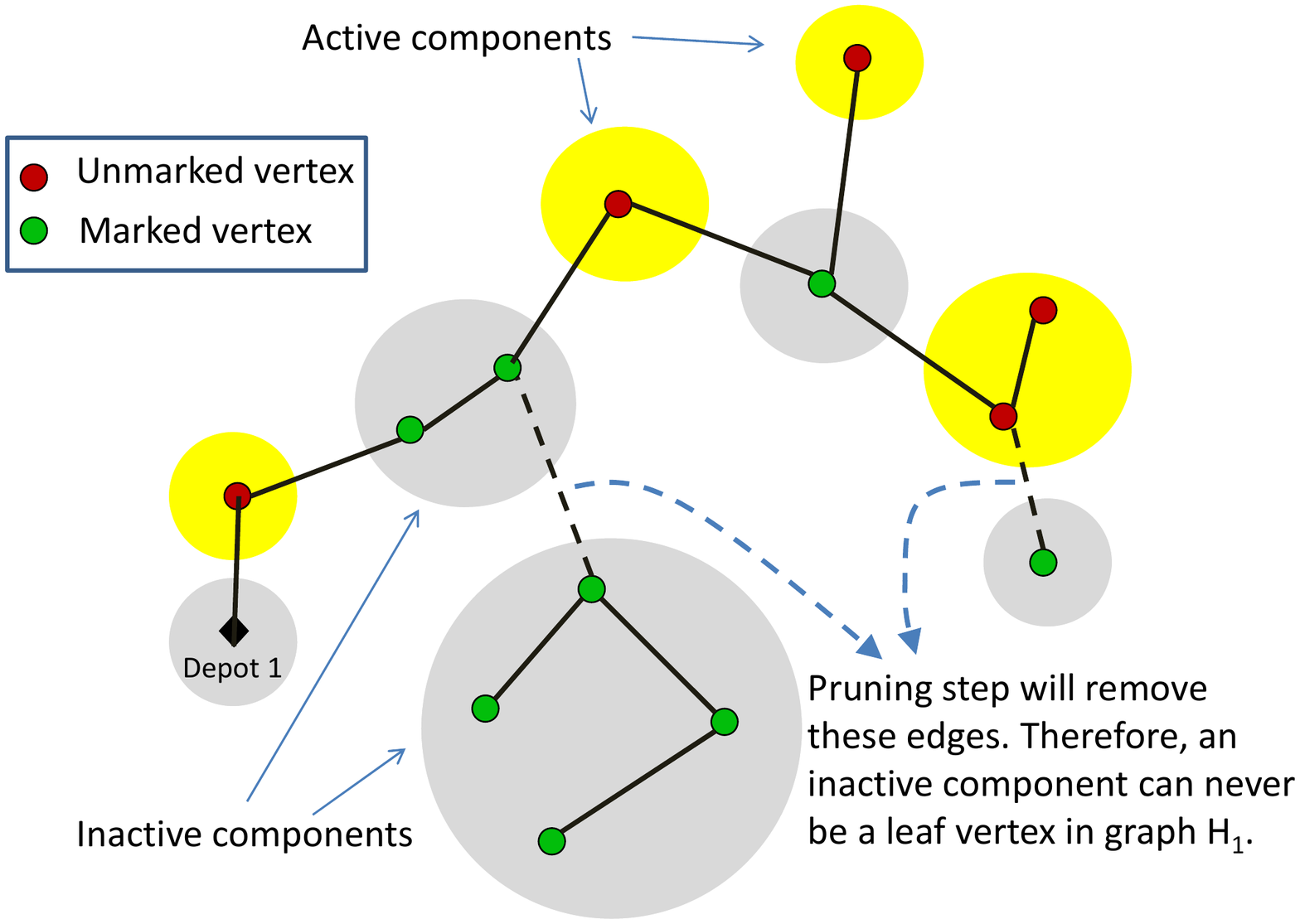}\\
  \caption{An example which illustrates that the graph $H_1$ cannot have an inactive component as its leaf vertex unless it contains $d_1$. The circles indicate all the active and the inactive components corresponding to the first vehicle at the start of an iteration. }
  \label{degreeH}
\end{figure}

We now claim that any vertex $v$ in $H_1$ that represents an inactive component in $N_d$ must have its degree $deg(v,H_1) \geq 2$ unless the inactive component contains the depot $d_1$. This result follows from the fact that a component, which does not contain $d_1$, can become inactive in ${\mathcal{C}}_1$ only if the constraint associated with this component in (\ref{eq:dualbound}) becomes tight. Therefore, all the vertices in this inactive component must be marked. Also, if vertex $v$ is a leaf ($deg(v,H_1)=1)$ then pruning all the edges from this inactive component will not disconnect any unmarked target from $d_1$. Hence, the pruning step of the algorithm will ensure that an inactive component can never be a leaf vertex in $H_1$ unless it contains $d_1$. Refer to figure \ref{degreeH} for an illustration of this claim. Hence, $\sum_{v\in N_d} deg(v,H_1) \geq 2|N_d|-1$. We now show the final part of the proof:

{
\begin{align}
& \sum_{v\in N_a} deg(v,H_1) +  \sum_{i=1}^m  \sum_{v\in M_{ai}} deg(v,H_{2i})  \\
= & \sum_{v\in N_a\bigcup N_d} deg(v,H_1) - \sum_{v\in N_d} deg(v,H_1)\\
 & + \sum_{i=1}^m  [\sum_{v\in M_{ai}\bigcup \{M_d\}} deg(v,H_{2i}) - deg(M_d,H_{2i})]\\
\leq & \sum_{v\in N_a\bigcup N_d} deg(v,H_1) - \sum_{v\in N_d} deg(v,H_1)\\
 & + \sum_{i=1}^m  [\sum_{v\in M_{ai}\bigcup \{M_d\}} deg(v,H_{2i})]\\
 \end{align}

$H_1$ is a tree that spans all the vertices in $N_a\bigcup N_d$. Therefore, the sum of the degree of all the vertices in $H_1$ is 2($|N_a| +|N_d| -1$). Similarly, $H_{2i}$ is a tree that spans all the vertices in $M_{ai}\bigcup \{M_d\}$. Therefore, the sum of the degree of all the vertices in $H_{2i}$ is 2$|M_{ai}|$. Hence, continuing with the proof,

\begin{align}
& \sum_{v\in N_a} deg(v,H_1) +  \sum_{i=1}^m  \sum_{v\in M_{ai}} deg(v,H_{2i})  \\
\leq & 2(|N_a|+|N_d|-1) - (2|N_d|-1) + 2\sum_{i=1}^m |M_{ai}| \\
< & 2|N_a|+ 2\sum_{i=1}^m |M_{ai}|.
\end{align}
}
Hence proved.
\end{proof}

\bibliographystyle{IEEEbib}
\bibliography{refer_thesis}

\end{document}